\documentclass[conference]{IEEEtran}
\usepackage{amsthm}
\usepackage{amssymb}

\usepackage{cite}

\ifCLASSINFOpdf
   \usepackage[pdftex]{graphicx}
\else
\fi
\usepackage{amsmath}
  \usepackage[caption=false,font=footnotesize]{subfig}
\begin{document}
%
\title{On Age of Information for Discrete Time Status Updating System With Ber/G/1/1 Queues}


\author{\IEEEauthorblockN{Jixiang Zhang and Yinfei Xu}
\IEEEauthorblockA{School of Information Science and Engineering\\
Southeast University, Nanjing 210096, China}
\IEEEauthorblockA{Email: \{230179077, yinfeixu\}@seu.edu.cn}
}


%


\maketitle

\newtheorem{Definition}{Definition}
\newtheorem{Theorem}{Theorem}
\newtheorem{Lemma}{Lemma}
\newtheorem{Corollary}{Corollary}

\begin{abstract}
In this paper, we consider the age of information (AoI) of a discrete time status updating system, focusing on finding the stationary AoI distribution assuming that the Ber/G/1/1 queue is used. Following the standard queueing theory, we show that by invoking a two-dimensional state vector which tracks the AoI and packet age in system simultaneously, the stationary AoI distribution can be derived by analyzing the steady state of the constituted two-dimensional stochastic process. We give the general formula of the AoI distribution and calculate the explicit expression when the service time is also geometrically distributed. The discrete and continuous AoI are compared, we depict the mean of discrete AoI and that of continuous time AoI for system with M/M/1/1 queue. Although the stationary AoI distribution of some continuous time single-server system has been determined before, in this paper, we shall prove that the standard queueing theory is still appliable to analyze the discrete AoI, which is even stronger than the proposed methods handling the continuous AoI.
\end{abstract}


%
\IEEEpeerreviewmaketitle

\section{Introduction}

Nowadays, real time data exchange and information transmission are getting more and more important, especially in those IoT applications. A network node needs fresh information to implement certain computations or make resource scheduling. The age of information (AoI) metric was proposed in \cite{1} to measure the freshness of the messages obtained at the receiver, which is defined as the time a packet experiences since it was generated at the source up to now. Since then, lots of papers have been published considering different aspects of AoI, both in theory and in practice \cite{2}.

In work \cite{3}, the author determined the limiting time average AoI for the updating system with M/M/1, M/D/1 and D/M/1 queues, respectively. The mean of the AoI for the basic queues using LCFS discipline were obtained in \cite{4}. The authors of paper \cite{5} considered the system with multiple sources and derived the average AoI of each source, using the newly proposed method which is called the Stochastic Hybrid System (SHS) analysis. For the case the system has two parallel transmitters, the closed-form expression of average AoI was determined in work \cite{6} by sophisticated random events analysis. Another related metric called the peak age of information (PAoI) was introduced in \cite{7}, which is defined as the peak value of the AoI just before it drops when every time a packet is received by the reveiver. The AoI and peak AoI analysis for three queue models was carried out in \cite{7}, including M/M/1/1, M/M/1/2 and M/M/1/2*, where in the last queue model the packet waiting in queue can be replaced constantly by the newly arriving packets, which is newer in time.

We observe that the analysis of the discrete time AoI was considered in paper \cite{8}, where the authors determined the mean of the AoI and peak AoI for system with Ber/G/1 and G/G/$\infty$ queues by the results obtained in continuous AoI analysis. After then, the AoI of discrete time status updating system was also analyzed in work \cite{9}. The authors gave a new description for one sample path of the AoI stochastic process. Provided these results, a general expression of the generation function of the AoI and peak was derived, which is equivalent to obtain their stationary distributions. Recently, using the queueing theory methods, the discrete time AoI and peak AoI distributions are numerically determined in series work \cite{10,11,12}. In \cite{13}, the stationary distribution of continuous AoI is determined for single-server system with various queue models, and it seems that the ideas and methods used in \cite{9} are from paper \cite{13}.

In this paper, we consider the discrete AoI of status updating system with Ber/G/1/1 queue, concentrate on determining the stationary distribution of the AoI at the destination. In \cite{9}, the authors have obtained the AoI distribution assuming that the queue used is Geo/Geo/1. Our work is different from theirs in two aspects. Firstly, in their model, the queue buffer is unbounded while we consider the bufferless system here. The AoI analysis for system having infinite size is \emph{different} from that of finite size system. Secondly, it is not clear whether the service time distribution can be relaxed to follow an arbitrary distribution in \cite{9}, since the AoI distribution of the generalized system is not given explicitly. While in this work, we show that by defining a two-dimensional state vector which records the instantaneous AoI and the age of the packet in system, all the random transitions can be described following standard queueing theory. Therefore, the steady state of the established two-dimensional process can give the stationary AoI, because the AoI is contained as part of the state vector. It should be pointed out that our basic idea is straightforward and different from the methods proposed and used in \cite{9} and \cite{13}.

The concept of multi-dimensional state vector is also introduced in\cite{5}, where it makes up the SHS analysis for continuous time status updating system. In fact, what we do here is to generalize SHS analysis to discrete time cases. Due to the difficulty of solving the system of rate (or intensity)-balance equations, only the mean or some simple moments of AoI can be calculated via the SHS analysis in \cite{5}. However, we will show that when the discrete AoI is considered, the system of probability-balance equations can be solved completely for some kinds of status updating systems, such that the stationary distribution of any state component in defined muiti-dimensional state vector can be obtained, not only the distribution of AoI.

The stationary AoI distribution of many single-server system has been determined by its Laplace-Stieltjes transformation (LST) in \cite{13}. Nevertheless, calculating the reverse transformation is not easy. The stationary distribution usually has no closed-form expression. Due to this, it is not easy to reduce the continuous results of AoI analysis to discrete versions by direct discretization. On the contrary, as long as the system of stationary equations are solved, we can get the explicit expression of discrete AoI distribution directly. Moreover, we can approach the continuous AoI distribution beginning with a approximate discrete time system, since any continuous probability distribution can be approximated by a $k$-point discrete one. Along this way, the AoI probability density function is established, avoiding performing the complex reverse LST transformation.

The rest of this paper is organized as follows. We describe the basic status updating system and give the necessary definition in Section II. The main results are given in Section III. We obtain the general formula of the stationary AoI distribution assuming that Ber/G/1/1 queue is used in system. As a specific example, AoI distribution expression of system with Ber/Geo/1/1 queue is calculated. The distribution curves for different system parameters are depicted in Section IV. In the concluding remark of Section V, we discuss how to generalize our analysis method to more AoI systms.


\section{System Model and Problem Formulation}
The basic status updating system consists of a source node and a destination node. The source observes a physical process $X(t)$ and samples the states of $X(t)$ at random times. An updating packet is composed of the sampled state $X(t_i)$ of the process and the sample time $t_i$. The source sends the packets to destination via a transmitter, which is modeled as a queue. In Figure \ref{fig1}, the model of a status updating system is depicted.

\begin{figure}[!t]
\centering
\includegraphics[width=3.3in]{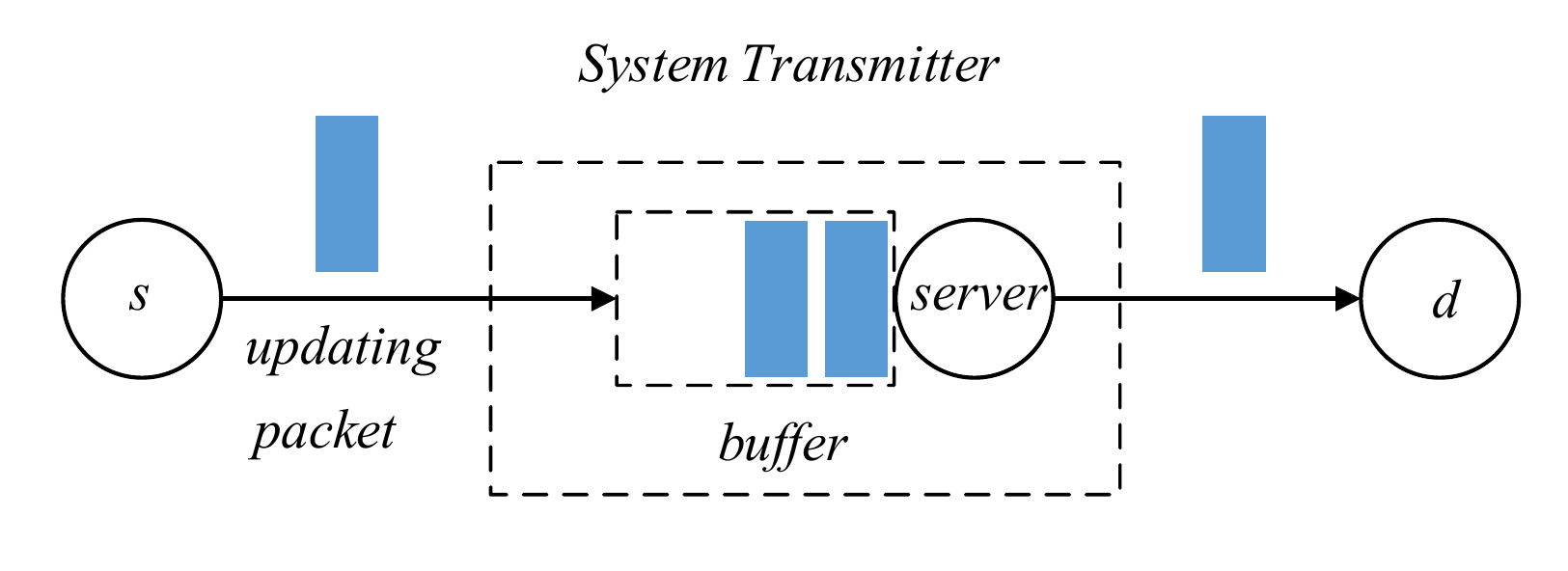}
\caption{The model of a status updating system.}
\label{fig1}
\end{figure}

In the discrete time model, both the time and the AoI are discretized. During the time when no packet is obtained at the receiver, the AoI increases constantly. When the receiver gets an update, the value of the AoI reduces to the system time of the arriving packet at the end of that time slot. In general, the system time of a packet equals the sum of the waiting time in queue and the service time the packet consumes at the server. However, for the queue model used in this paper, the packet waiting time equals zero.

Now, we give the definition of the time average AoI under the discrete time setting.
\begin{Definition}
In the discrete time model, the limiting time average age of information $\overline{\Delta}$ is defined as
\begin{equation}
\overline{\Delta}=\lim_{T\to\infty}\frac{1}{T}\sum\nolimits_{k=1}^{T}a(k)
\end{equation}
where $a(k)$ denotes the value of AoI at $k$th time slot.
\end{Definition}

We simply deal with the expression (1) as follows.
\begin{align}
\overline{\Delta}&=\lim_{T\to\infty}\frac{1}{T}\sum\nolimits_{k=1}^{T}a(k) \notag \\
&=\lim_{T\to\infty}\frac{1}{T}\sum\nolimits_{n=1}^{M}n\cdot |\{k:a(k)=n\}| =\sum\nolimits_{n=1}^{\infty} n\cdot \pi_n \notag
\end{align}
where
\begin{equation}
\pi_n=\lim_{T\to\infty} \frac{|\{k:a(k)=n\}|}{T} \qquad (n \geq 1)    \notag
\end{equation}
is the probability that the AoI takes value $n$ when the observing time tends to infinity. The number $M$ is defined to be $M=\text{max}_{1\leq k \leq T}a(k)$.

Always the AoI processes are assumed to be ergodic. Ergodicity property ensures that we can obtain the AoI performance metrics by investigating any one sample path of the process. Notice that if all the probabilities $\pi_n$, $n \geq 1$ are determined, then the distribution of the AoI is known as well. 

\section{Discrete AoI-distribution: Ber/G/1/1 queues}
In this Section, assume that the queue model is Ber/G/1/1, we derive the general expression of the stationary AoI distribution.

\subsection{Stationary distribution of the AoI: Ber/G/1/1 queues}
Assume that at each time slot, a new updating packet comes independently and with identical probability $p$. We use $A$ to denote the interarrival time between successive arriving packets, then the distribution of $A$ is
\begin{equation}
\Pr\{A=j\}=(1-p)^{j-1}p  \qquad  (j\geq 1)   \notag
\end{equation}

The packet service time is represented by $B$, suppose that its distribution is written as
\begin{equation}
\Pr\{B=j\}=q_j  \qquad  (j\geq 1)   \notag
\end{equation}

The updating packet comes at the beginning of one time slot, while the AoI is recorded and modified at the end of each time slot. The terms state, age-state and the state vector are used interchangeably.

Define the two-dimensional vector $(n_k,m_k)$, where $n_k$ denotes the value of the AoI at the $k$th time slot and the second component $m_k$ represents the age of the packet in system at that time. Constituting the two-dimensional stochastic process
\begin{equation}
AoI_{Ber/G/1/1}=\left\{(n_k,m_k): n_k>m_k \geq 0,k \geq 1\right\} \notag
\end{equation}

In the following paragraphs, we analyze all the random transfers of age-states $(n,m)$ and determine the transtion probabilities, such that the stationary equations describing the steady state of the AoI process can be established. Notice that the first component of the two-dimensional state vector, $n$, denotes the AoI at the destination. Therefore, as long as all the stationary probabilities of state $(n,m)$ are determined, the distribution of the AoI is obtained either. An illustration of the discrete state and state transfers are provided in Figure \ref{fig2}.

\begin{figure}[!t]
\centering
\includegraphics[width=3.2in]{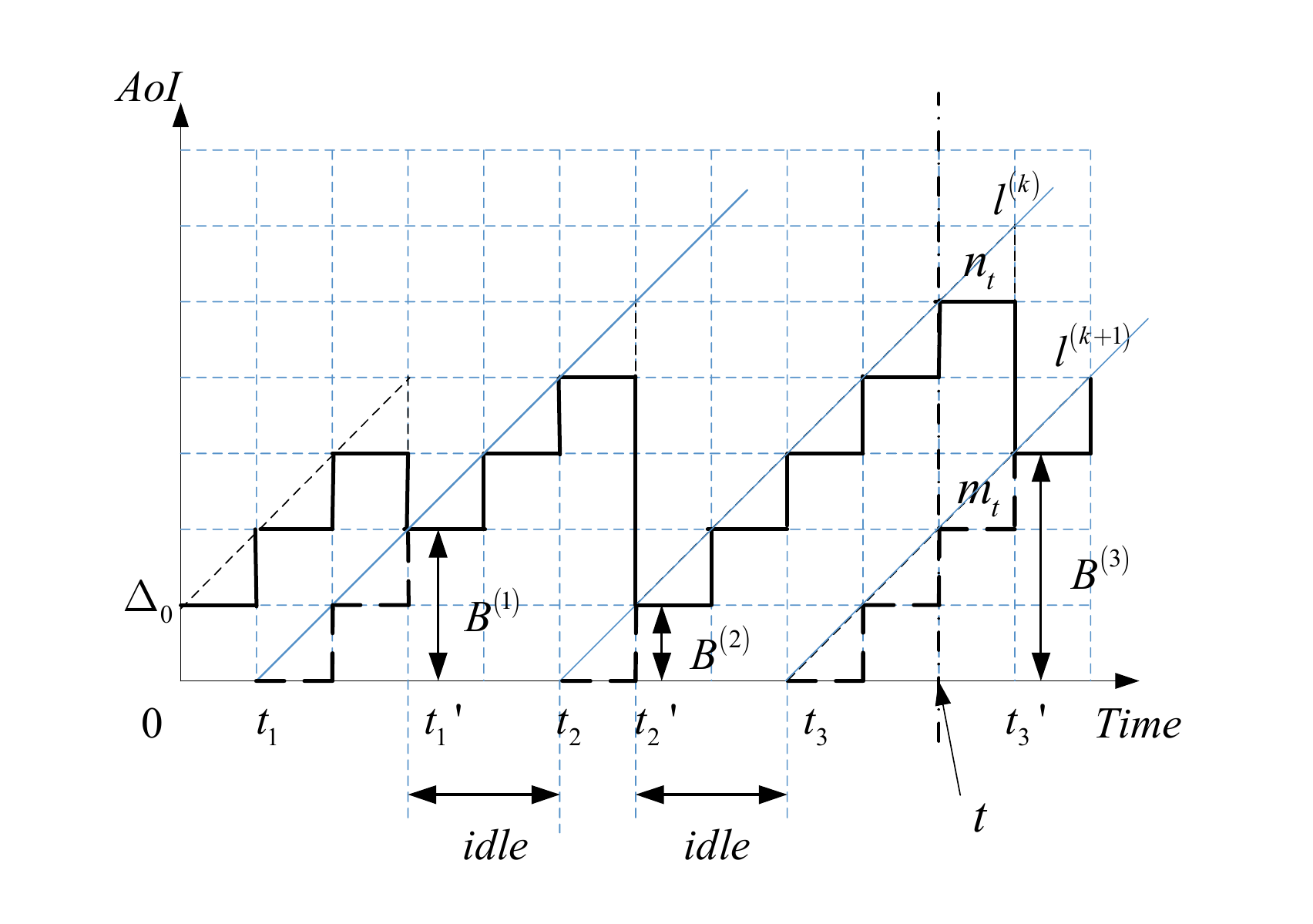}
\caption{An illustration of discrete state and state transfers: Status updating system with Ber/G/1/1 queue.}
\label{fig2}
\end{figure}

First of all, assume that the state vector is $(n,m)$ at the $k$th time slot, where $n>m>0$. At the $(k+1)$th time slot, it shows that $(n,m)$ can only transfer to either $(n+1,m+1)$ or $(m+1,0)$ depending on whether or not the packet service is complete. Because $m$ is greater than zero, which means an updating packet is served at the transmitter currently and the system is full. Let $P_{(n,m),(i,j)}$ be the transition probability the age-state jumps to $(i,j)$ at the next time slot from the current state $(n,m)$. We have that
\begin{equation}
P_{(n,m),(n+1,m+1)}=\Pr\{B>m | B>m-1\}
\end{equation}

The conditional probability is used in equation (2) because when the service time $B$ is longer than $m$ time slots, we must know the event $\{B>m-1\}$.

On the contrary, if the service of the packet finishes at the next time slot, the system is emptied and the age-state changes to $(m+1,0)$. Therefore, we have
\begin{equation}
P_{(n,m),(m+1,0)}=\Pr\{B=m | B>m-1\}
\end{equation}

Next, consider the case $m=0$. At this time, the transmitter is idle and the system is empty. Thus, the state transitions are totally governed by the interarrival time $A$.

For the state vector $(n,0)$, $n \geq 1$, if no packet comes at next time slot, then
\begin{equation}
P_{(n,0),(n+1,0)}=1-p
\end{equation}

Assume that a new packet comes at next time slot. In this case, we have to consider two sub-cases further. When the packet comes and does not leave after one time slot, the age-state changes to $(n+1,1)$. Otherwise, if the arriving packet finishes the service in single time slot and gets to receiver, we show that the state vector will jump to $(1,0)$. Summarize above discussions, we obtain
\begin{equation}
P_{(n,0),(n+1,1)}=p\Pr\{B>1\}, \quad  P_{(n,0),(1,0)}=p\Pr\{B=1\} \notag
\end{equation}

So far, all the state transfers are analyzed and the corresponding transition probabilities are determined. Define $\pi_{(n,m)}$ as the stationary probability of the state $(n,m)$. Then, we give the stationary equations of the AoI process $AoI_{Ber/G/1/1}$.
\begin{Theorem}
Denote $\pi_{(n,m)}$ as the stationary probability of the age-state $(n,m)$, where $n> m\geq 0$. Then, for the two-dimensional stochastic process $AoI_{Ber/G/1/1}$, the stationary equations are determined as
\begin{equation}
\begin{cases}
\pi_{(n,m)}=\pi_{(n-1,m-1)}\Pr\{B>m-1 | B>m-2\}    \\
\qquad \qquad \qquad  \qquad \qquad \qquad \quad \quad \quad     (n >m \geq 2)   \\
\pi_{(n,1)}=\pi_{(n-1,0)}p \Pr\{B>1\}       \qquad \quad \quad    (n \geq 2)    \\
\pi_{(n,0)}=\pi_{(n-1,0)}(1-p)+ \left(\sum\nolimits_{k=n}^{\infty} \pi_{(k,n-1)}\right)   \\
\qquad \quad  \times  \Pr\{B=n-1| B>n-2 \}    \qquad     (n \geq 2)   \\
\pi_{(1,0)}=  \left( \sum\nolimits_{k=1}^{\infty}  \pi_{(k,0)} \right) p \Pr\{B=1\}
\end{cases}
\end{equation}
\end{Theorem}

\begin{proof}
We explain each line of (5) as follows. First of all, for the case $n>m \geq 2$, the system is full and is still full before one time slot. The newly arriving packet cannot enter the transmitter if there exists such one, so that the state transitions are totally determined by the random service time $B$. Assume that the current state is $(n-1,m-1)$, then with probability $\Pr\{B>m-1|B>m-2\}$ the age-state jumps to $(n,m)$ at next time slot. This gives the first line of (5).

Consider the state $(n,1)$ in second line, notice that before one time slot the system is empty because $m$ reduces to 0. Starting with $(n-1,0)$, let a packet comes and stays at the system, which occurs with probability $p \Pr\{B>1\}$, we show that the age-state transfers to $(n,1)$ at next time slot.

Now, the state transitions of age-state $(n,0)$, $n\geq2$ are analyzed.
From the state $(n-1,0)$, as long as no packet arrives in next time slot, the value of the AoI at the receiver increases 1 and the second component $m$ remains zero, we get the state vector $(n,0)$. On the other hand, assume that the original state is $(k,n-1)$ and the service of the packet finishes at the next time slot, it shows that the state will transfer to $(n,0)$ as well. Combining all the transitions from $(k,n-1)$ to $(n,0)$ where $k \geq n$, we obtain the stationary equations for state vectors $(n,0)$, $n \geq 2$ and explain the third line of (5).

Finally, beginning with an empty system, an arriving packet experiences single time slot service and then delivers to destination will make the age-state jump to $(1,0)$, which yields the last equation in (5).
\end{proof}

As long as all the probabilities $\pi_{(n,m)}$ are determined by solving the system of equations (5), the stationary distribution of the AoI can be obtained since we have
\begin{equation}
\Pr\{\Delta = n \}=\sum\nolimits_{m=0}^{n-1} \pi_{(n,m)} \qquad (n \geq 1)
\end{equation}

Furthermore, for $k\geq1$, the distribution function of the AoI can also be obtained as
\begin{equation}
\Pr\{\Delta \leq k\}=\sum\nolimits_{n=1}^{k} \Pr\{\Delta=n\}=\sum\nolimits_{n=1}^{k} \sum\nolimits_{m=0}^{n-1}\pi_{(n,m)} \notag
\end{equation}

We solve the system of equations (5) and determine all the stationary probabilities $\pi_{(n,m)}$ in Theorem 2.

\begin{Theorem}
When the AoI process $AoI_{Ber/G/1/1}$ reaches the steady state, the stationary probabilities for all the age-states $(n,m)$ are given as follows. Firstly,
\begin{equation}
\pi_{(n,0)}=\frac{pq_1F(p,n)}{1+p(1-q_1)\sum\nolimits_{m=1}^{\infty}\sum\nolimits_{l=m}^{\infty}q_l} \quad (n\geq1)
\end{equation}
and the probabilities $\pi_{(n,m)}$, $n>m\geq1$ are solved as
\begin{equation}
\pi_{(n,m)}=\frac{p^2q_1(1-q_1)F(p,n-m)\left(\sum\nolimits_{l=m}^{\infty}q_l \right)}{1+p(1-q_1)\sum\nolimits_{m=1}^{\infty}\sum\nolimits_{l=m}^{\infty}q_l}
\end{equation}
where for $n\geq2$, define that
\begin{align}
F(p,n)=(1-p)^{n-1} +\frac{1-q_1}{q_1}\sum\nolimits_{j=0}^{n-2}(1-p)^{j}q_{n-1-j}
\end{align}
and denote $F(p,1)=1$.
\end{Theorem}

The proof of Theorem 2 is given in Appendix A.

Given all the stationary probabilities, the distribution of AoI can be determined by equation (6).
\begin{Corollary}
Assume that the queue used in status updating system is Ber/G/1/1. When the system runs in steady state, the stationary AoI-distribution is calculated as
\begin{align}
&\Pr\{\Delta=n\}=\pi_{(n,0)}+ \sum\nolimits_{m=1}^{n-1}\pi_{(n,m)} \notag \\
={}&\frac{pq_1F(p,n)+p^2q_1(1-q_1)\sum\nolimits_{m=1}^{n-1}F(p,n-m)\left(\sum\nolimits_{l=m}^{\infty}q_l\right)}{1+p(1-q_1)\sum\nolimits_{m=1}^{\infty}\sum\nolimits_{l=m}^{\infty}q_l}   \notag
\end{align}

The AoI cumulative probabilities are determined as
\begin{equation}
\Pr\{\Delta \leq k \} =\sum\nolimits_{n=1}^{k}\Pr\{\Delta=n\} \quad (k\geq1) \notag
\end{equation}
where $F(p,n)$ is defined in equation (9).
\end{Corollary}

\begin{proof}
Since the first component of the state vector denotes the AoI at the receiver, the probability that the AoI equals $n$ can be obtained by adding up all the stationary probabilities having $n$ as first component. We show that
\begin{align}
&\Pr\{\Delta=n\}=\sum\nolimits_{m=0}^{n-1} \pi_{(n,m)}   \notag \\
={}& \pi_{(n,0)} + \sum\nolimits_{m=1}^{n-1} \pi_{(n,m)}   \notag \\
={}& \frac{pq_1F(p,n)}{1+p(1-q_1)\sum\nolimits_{m=1}^{\infty}\sum\nolimits_{l=m}^{\infty}q_l} \notag \\
  {}& +\sum\nolimits_{m=1}^{n-1} \frac{p^2q_1(1-q_1)F(p,n-m)\left(\sum\nolimits_{l=m}^{\infty}q_l\right)}{1+p(1-q_1)\sum\nolimits_{m=1}^{\infty}\sum\nolimits_{l=m}^{\infty}q_l}  \notag \\
={}&\frac{pq_1F(p,n)+p^2q_1(1-q_1)\sum\nolimits_{m=1}^{n-1}F(p,n-m)\left(\sum\nolimits_{l=m}^{\infty}q_l\right)}{1+p(1-q_1)\sum\nolimits_{m=1}^{\infty}\sum\nolimits_{l=m}^{\infty}q_l} \notag
\end{align}

For $k\geq1$, the cumulative probability of the AoI equals the sum of probabilities the AoI takes value $n$ from $n=1$ to $k$. This completes the proof of Corollary 1.
\end{proof}

\subsection{The AoI stationary distribution: Ber/Geo/1/1 queue}

When the packet service time $B$ is also a geometric random variable, we can calculate the explicit expression of AoI distribution. Suppose that $B$ is distributed as
\begin{equation}
\Pr\{B=k\}=(1-\gamma)^{k-1}\gamma  \quad (k \geq 1)
\end{equation}

To keep queueing system stable, let $p < \gamma$. We first find all the stationary probabilities by solving the system of equations (5), and then derive the explicit expression of AoI-distribution. The results are stated in Theorem 3 below.
\begin{Theorem}
When the service time $B$ is also geometrically distributed, the stationary probabilities $\pi_{(n,m)}$ are solved as
\begin{equation}
\begin{cases}  \notag
\pi_{(n,0)}=\frac{p\gamma ^2\left[(1-p)^n - (1-\gamma)^n \right] }{(p+\gamma -p\gamma)(\gamma -p)} &   (n\geq 1) \\
\pi_{(n,m)}=\frac{(p\gamma )^2\left[(1-p)^{n-m}(1-\gamma)^{m}-(1-\gamma)^n \right]}{(p+\gamma -p\gamma)(\gamma -p)}& (n>m \geq 1)
\end{cases}
\end{equation}
\end{Theorem}

The Theorem 3 can be proved either from the general formulas (7)-(9) or solving system of equation (5) directly. We give the derivation details in Appendix B.

By collecting all the stationary probabilities with identical AoI-component, we can derive the stationary distribution of the AoI. The cumulative probability distribution of the AoI is calculated accordingly. The calculation detalis are omitted.

\begin{Corollary}
The distribution of AoI for the status updating system with Ber/Geo/1/1 queue is given as
\begin{align}
\Pr\{\Delta=n\}&= \frac{p(1-p)\gamma ^3 \left[(1-p)^n-(1-\gamma)^n \right]}{(p+\gamma -p\gamma)(\gamma -p)^2}  \notag \\
&\qquad -\frac{(p\gamma)^2 n(1-\gamma)^n}{(p+\gamma -p\gamma)(\gamma -p)} \qquad (n \geq 1)
\end{align}
and the AoI cumulative probability is determined as
\begin{align}
\Pr\{\Delta \leq k\}&= 1-\frac{(1-p)\gamma^2\left[(1-p)^{k+1}\gamma -p(1-\gamma)^{k+1} \right]}{(p+\gamma -p\gamma)(\gamma -p)^2} \notag \\
& \quad +  \frac{p^2 (1+k\gamma)(1-\gamma)^{k+1}}{(p+\gamma -p\gamma)(\gamma -p)} \qquad (k\geq1)
\end{align}
\end{Corollary}

\begin{figure*}[!t]
\centering
\subfloat[Stationary distribution of discrete AoI]{\includegraphics[width=2.2in]{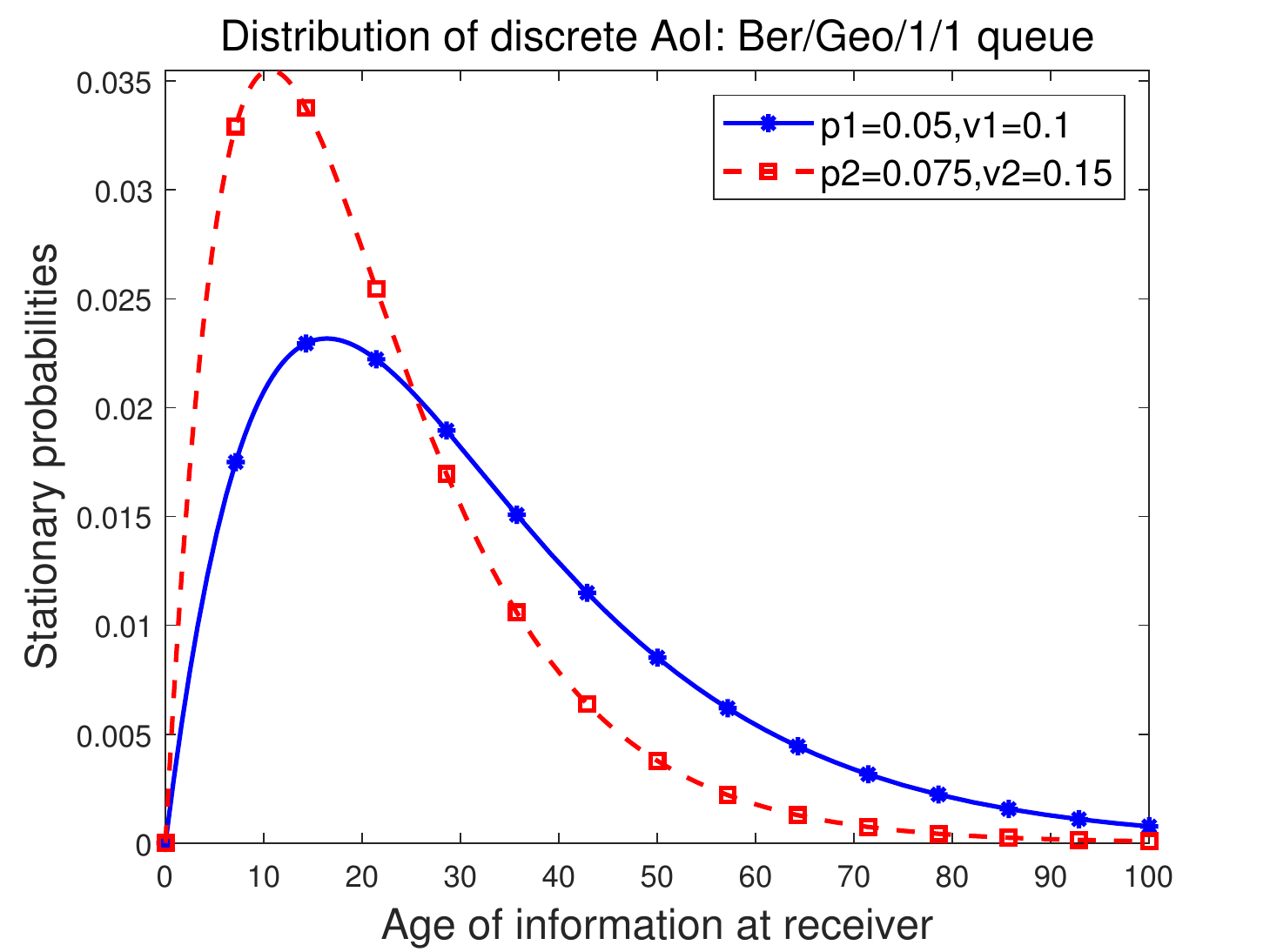}%
\label{fig_fik1}}
\hfil
\subfloat[Discrete AoI cumulative probabilities]{\includegraphics[width=2.2in]{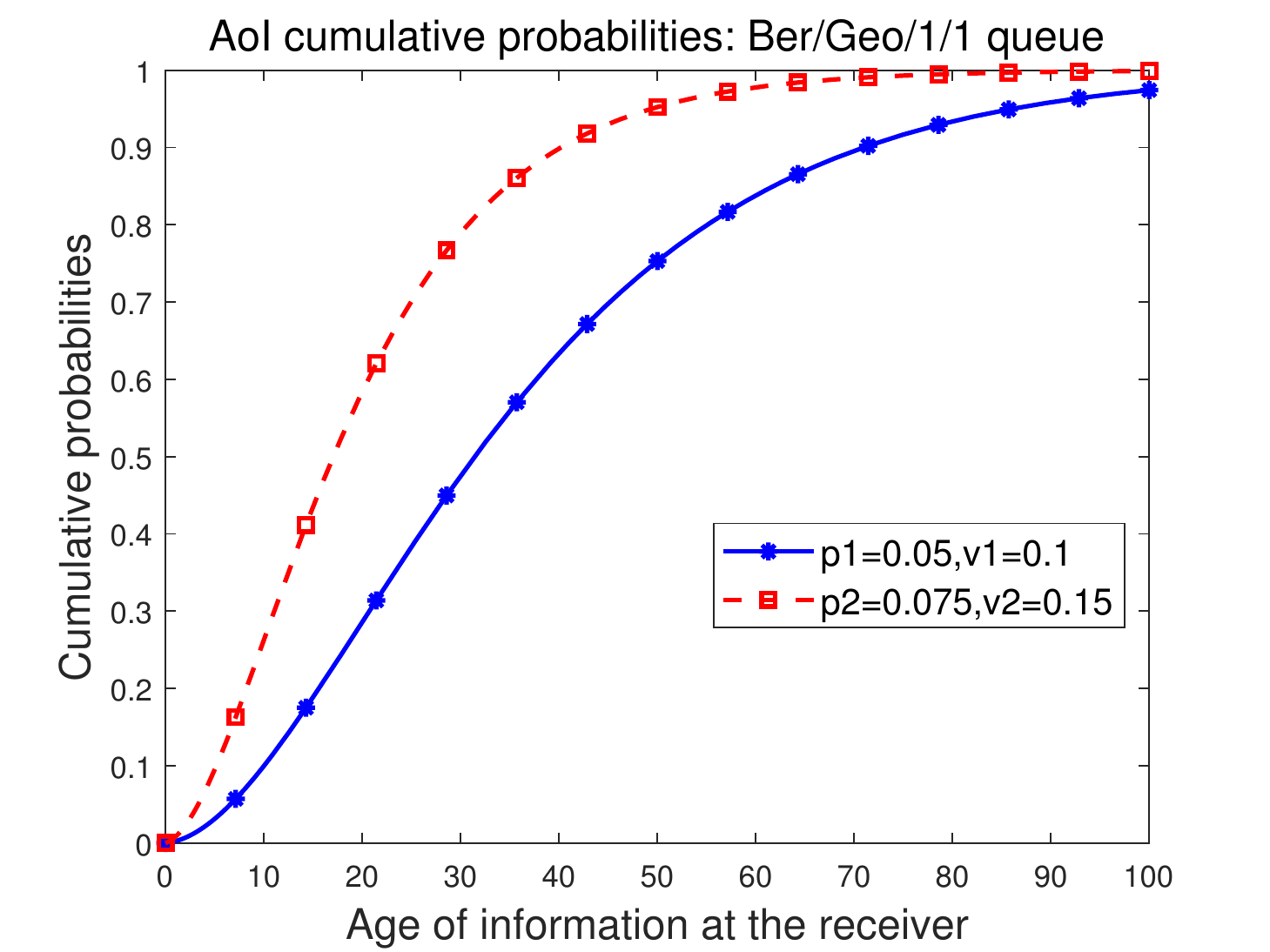}%
\label{fig_fik2}}
\hfil
\subfloat[Average continuous and discrete AoI]{\includegraphics[width=2.2in]{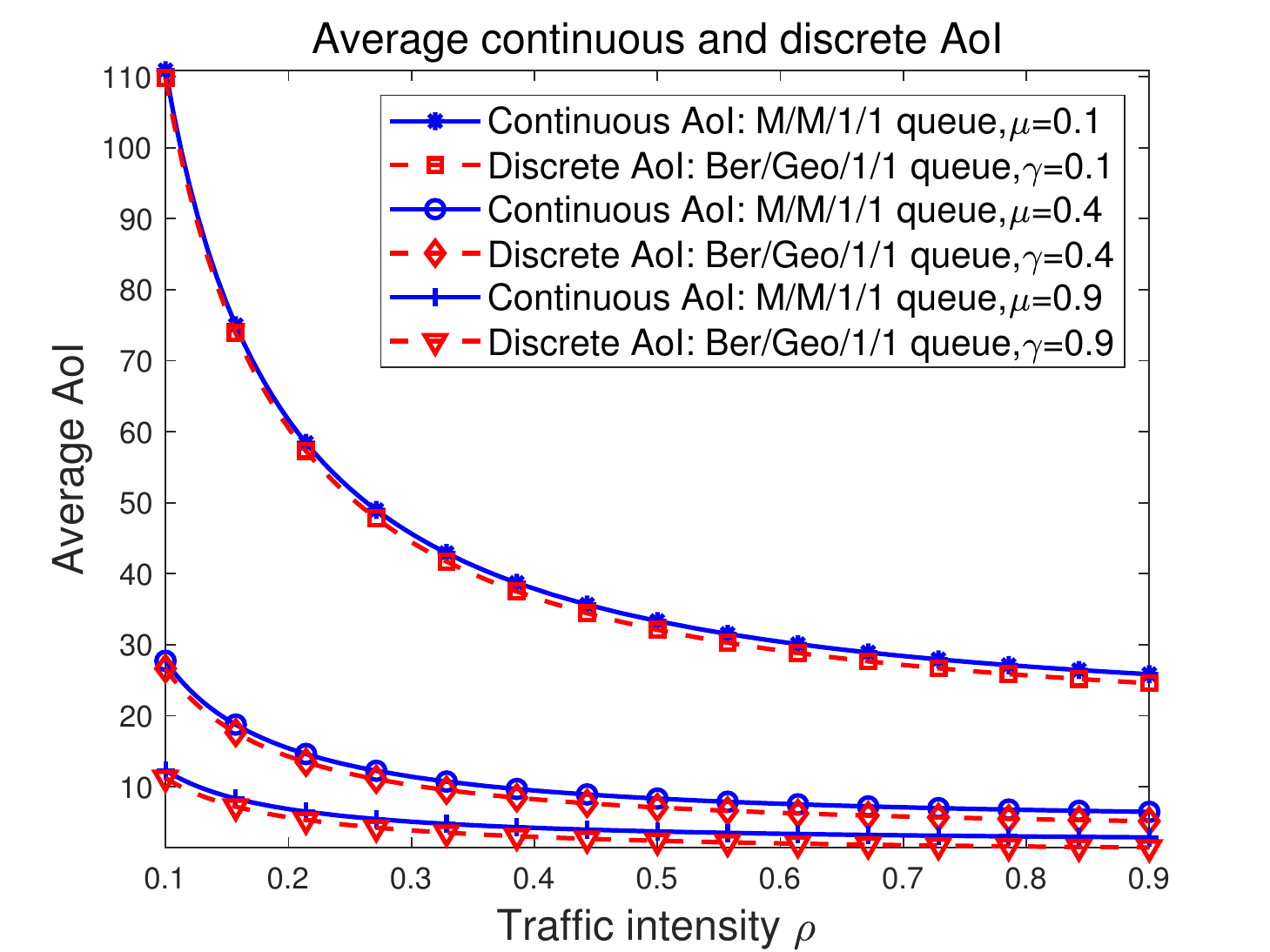}%
\label{fig_fik3}}
\caption{Numerical simulations of discrete AoI and discrete approximations of mean of continuous AoI.}
\label{fig3}
\end{figure*}

Notice that in order to obtain the stationary distribution of AoI, we invoke a two-dimensional state vector $(n,m)$ where the first component denotes the instantaneous AoI at the destination, while the latter parameter represents the age of the packet in system. Since the system size is 1, so that the packet waiting time is zero, the second component $m$ is exactly the packet service time, which should be geometrically distributed in this case. We prove this claim by computing the marginal distribution of $M$, which is defined as the random variable of the second component of the age-state.

For $m\geq 1$, we show that
\begin{align}
&\Pr\{M=m\}=\sum\nolimits_{n=m+1}^{\infty}\pi_{(n,m)}  \notag \\
={}& \sum\nolimits_{n=m+1}^{\infty}\frac{(p\gamma)^2\left[(1-p)^{n-m}(1-\gamma)^m-(1-\gamma)^n \right] }{(p+\gamma-p\gamma)(\gamma-p)} \notag \\
={}& \frac{(p\gamma)^2}{(p+\gamma-p\gamma)(\gamma-p)}\left[\frac{1-p}{p}(1-\gamma)^m - \frac{1-\gamma}{\gamma}(1-\gamma)^m \right]  \notag \\
={}& \frac{p\gamma}{p+\gamma-p\gamma}(1-\gamma)^m
\end{align}

We have used the probability expressions obtained in Theorem 3. Probability (13) should be normalized by the probability that the transmitter is not idle, since in that case we have $m=0$. It shows that
\begin{align}
&\Pr\{\text{system is not idle}\} \notag \\
={}&1-\sum\nolimits_{n=1}^{\infty}\pi_{(n,0)} =\frac{p(1-\gamma)}{p+\gamma-p\gamma}  \notag
\end{align}

Therefore, for $m\geq1$, we obtain that
\begin{equation}
\Pr\{B=m\}=\frac{\Pr\{M=m\}}{\Pr\{\text{system is not idle}\}}=(1-\gamma)^{m-1}\gamma  \notag
\end{equation}
which is indeed the geometric distribution with parameter $\gamma$.

From (11), the mean of the discrete AoI can be calculated. We will show that when the time slot is sufficiently small, the average discrete AoI is close to average continuous AoI, which is denoted by $\overline{\Delta}_{M/M/1/1}$ and is determined in \cite{7} as
\begin{equation}
\overline{\Delta}_{M/M/1/1}=\frac{1}{\mu}\left(1+\frac{1}{\rho}+\frac{\rho}{1+\rho}\right)  \notag
\end{equation}

\begin{Corollary}
For the system with Ber/Geo/1/1 queue, the limiting time average discrete AoI is equal to
\begin{equation}
\overline{\Delta}^{(d)}_{Ber/Geo/1/1}=\frac{1}{\gamma}\left((1-\gamma) + \frac{1}{\rho_d} + \frac{\rho_d}{\frac{1}{1-\gamma}+\rho_d}\right)
\end{equation}
and approaches $\overline{\Delta}_{M/M/1/1}$ in almost all the range of $\rho_d$, which is discrete traffic intensity and is defined as $\rho_d=p/\gamma$.
\end{Corollary}

The Proof of Corollary 3 is shown in Appendix C.

%

\section{Numerical Results}
In Figure \ref{fig3}, we draw the stationary distribution of discrete AoI along with its cumulative probabilities. In addition, the mean of discrete and continuous AoI are also depicted.

We take two group of parameters $(p,\gamma)$ with fixed ratio 0.5. From Figure \ref{fig_fik1}, it sees that the AoI distribution curve is more concentrated when $p$ and $\gamma$ are both larger, and has higher peak probability. The cumulative probability distributions also show that when the parameters are larger, the AoI cumulative probability gets to 1 more faster.

Average AoIs of system with M/M/1/1 and Ber/Geo/1/1 queue are given in Figure \ref{fig_fik3}. We draw both curves for three different cases, where the service rate $\mu$ (or $\gamma$) is fixed and the traffic intensity varies. Numerical results show that in almost all the ranges of $\rho$ (or $\rho_d$), the mean of discrete AoI is very close to its continuous counterpart.



\section{Conclusion}
In this paper, we find the stationary distribution of discrete AoI for the system with Ber/G/1/1 queue. By defining a two-dimensional age-state which contains AoI and describing all the random state transitions, we show that the AoI distribution can be obtained if the steady state of constituted two-dimensional stochastic process is solved. The idea of using a multiple-dimensional state vector comes from SHS analysis of AoI. What we do is generalizing SHS method to discrete time status updating systems. The numerical results show that the mean of discrete AoI is very close to the mean of continuous AoI when time slot is small enough.

To obtain the exact probability density function of continuous AoI, often a reverse LST transformation has to be computed which is not easy. We point out that although the proposed idea is straightforward, it is still appliable to analyze AoI of more systems. The biggest advantage is that the exact expression of AoI distribution can be obtained as long as the steady state of constituted AoI stochastic process is solved.

\bibliographystyle{IEEEtran}
\bibliography{paperlists}

\begin{thebibliography}{10}
\providecommand{\url}[1]{#1}
\csname url@samestyle\endcsname
\providecommand{\newblock}{\relax}
\providecommand{\bibinfo}[2]{#2}
\providecommand{\BIBentrySTDinterwordspacing}{\spaceskip=0pt\relax}
\providecommand{\BIBentryALTinterwordstretchfactor}{4}
\providecommand{\BIBentryALTinterwordspacing}{\spaceskip=\fontdimen2\font plus
\BIBentryALTinterwordstretchfactor\fontdimen3\font minus
  \fontdimen4\font\relax}
\providecommand{\BIBforeignlanguage}[2]{{%
\expandafter\ifx\csname l@#1\endcsname\relax
\typeout{** WARNING: IEEEtran.bst: No hyphenation pattern has been}%
\typeout{** loaded for the language `#1'. Using the pattern for}%
\typeout{** the default language instead.}%
\else
\language=\csname l@#1\endcsname
\fi
#2}}
\providecommand{\BIBdecl}{\relax}
\BIBdecl

\bibitem{1}
S.~Kaul, M.~Gruteser, V.~Rai, and J.~Kenney, ``Minimizing age of information in
  vehicular networks,'' \emph{8th Annual IEEE Communications Society Conference
  on Sensor, Mesh and Ad Hoc Communications and Networks}, pp. 350--358, 2011.

\bibitem{2}
R.~D. Yates, Y.~Sun, D.~R.~B. III, S.~K. Kaul, E.~Modiano, and S.~Ulukus, ``Age
  of information: An introduction and survey,'' \emph{arXiv:2007.08564v1},
  2020.

\bibitem{3}
S.~Kaul, R.~Yates, and M.~Gruteser, ``Real-time status: How often should one
  update?'' \emph{2012 Proceedings IEEE INFOCOM}, pp. 2731--2735, 2012.

\bibitem{4}
------, ``Status updates through queues,'' \emph{46th Annual Conference on
  Information Sciences and Systems (CISS)}, 2012.

\bibitem{5}
R.~D. Yates and S.~Kaul, ``The age of information: Real-time status updating by
  multiple sources,'' \emph{IEEE Transactions on Information Theory}, vol.~65,
  no.~3, pp. 1807--1827, 2019.

\bibitem{6}
C.~Kam, S.~Kompella, G.~D. Nguyen, and A.~Ephremides, ``Effect of message
  transmission path diversity on status age,'' \emph{IEEE Transactions on
  Information Theory}, vol.~62, no.~3, pp. 1360--1374, 2016.

\bibitem{7}
M.~Costa, M.~Codreanu, and A.~Ephremides, ``Age of information with packet
  management,'' \emph{2014 IEEE International Symposium on Information Theory},
  pp. 1583--1587, 2014.

\bibitem{8}
V.~Tripathi, R.~Talak, and E.~Modiano, ``Age of information for discrete time
  queues,'' \emph{arXiv:1901.10463v1}, 2019.

\bibitem{9}
A.~Kosta, N.~Pappas, A.~Ephremides, and V.~Angelakis, ``Non-linear age of
  information in a discrete time queue: Stationary distribution and average
  performance analysis,'' \emph{ICC 2020 - 2020 IEEE International Conference
  on Communications (ICC)}, pp. 1--6, 2020.

\bibitem{10}
N.~Akar, O.~Doğan, and E.~U. Atay, ``Finding the exact distribution of (peak)
  age of information for queues of ph/ph/1/1 and m/ph/1/2 type,'' \emph{IEEE
  Transactions on Communications}, vol.~68, no.~9, pp. 5661--5672, 2020.

\bibitem{11}
N.~Akar and O.~Dogan, ``Discrete-time queueing model of age of information with
  multiple information sources,'' \emph{arXiv:2007.11650v1}, 2020.

\bibitem{12}
O.~Dogan and N.~Akar, ``The multi-source preemptive m/ph/1/1 queue with packet
  errors: Exact distribution of the age of information and its peak,''
  \emph{arXiv:2007.11656v1}, 2020.

\bibitem{13}
Y.~Inoue, H.~Masuyama, T.~Takine, and T.~Tanaka, ``A general formula for the
  stationary distribution of the age of information and its application to
  single-server queues,'' \emph{IEEE Transactions on Information Theory},
  vol.~65, no.~12, pp. 8305--8324, 2019.

\end{thebibliography}


\begin{thebibliography}{1}

\bibitem{IEEEhowto:kopka}
H.~Kopka and P.~W. Daly, \emph{A Guide to \LaTeX}, 3rd~ed.\hskip 1em plus
  0.5em minus 0.4em\relax Harlow, England: Addison-Wesley, 1999.

\end{thebibliography}
%


\newpage
\onecolumn

\appendices
\section{Proof of Theorem 2}
In this Appendix, we solve the system of equations (5) and derive all the stationary probabilities $\pi_{(n,m)}$. For convenience, the stationary equations are copied in the following.
\begin{equation}
\begin{cases}
\pi_{(n,m)}=\pi_{(n-1,m-1)}\Pr\{B>m-1 | B>m-2\}   &   (n >m \geq 2)   \\
\pi_{(n,1)}=\pi_{(n-1,0)}p \Pr\{B>1\} &   (n \geq 2)    \\
\pi_{(n,0)}=\pi_{(n-1,0)}(1-p)+ \left(\sum\nolimits_{k=n}^{\infty} \pi_{(k,n-1)}\right)  \Pr\{B=n-1| B>n-2 \}  &  (n \geq 2)   \\
\pi_{(1,0)}=  \left( \sum\nolimits_{k=1}^{\infty}  \pi_{(k,0)} \right) p \Pr\{B=1\}
\end{cases}
\end{equation}

Firstly, we show that the equations (17) can be simplified by eliminating the dependence on the second component $m$. For a general age-state $(n,m)$, $n>m \geq 2$, iteratively applying the first line of (17) yields
\begin{align}
\pi_{(n,m)}&=\pi_{(n-1,m-1)}\Pr\{B>m-1|B>m-2\}  \notag \\
&=\pi_{(n-2,m-2)}\Pr\{B>m-2|B>m-3\} \Pr\{B>m-1|B>m-2\}  \notag \\
&\qquad \qquad \qquad \qquad  \qquad \qquad  \quad \vdots   \notag \\
&=\pi_{(n-m+1,1)}\Pr\{B>1\} \times \cdots \times \Pr\{B>m-1|B>m-2\}  \notag \\
&=\pi_{(n-m+1,1)}\Pr\{B>m-1\} \notag \\
&=\pi_{(n-m,0)}p (1-q_1) \Pr\{B>m-1\}
\end{align}

For the last step recursion, we use the second line of (17).

Equation (18) shows that probability $\pi_{(n,m)}$, $n>m\geq 1$ can be represented by $\pi_{(n,0)}$, $n\geq1$. As a result, to solve system of equations (17) we only need to determine those stationary probabilities $\pi_{(n,0)}$.

Next, the infinite sum in the third line of (17) is calculated. We show that
\begin{align}
&\left(\sum\nolimits_{k=n}^{\infty} \pi_{(k,n-1)} \right) \Pr\{B=n-1|B>n-2\} \notag \\
={}&\left(\sum\nolimits_{k=n}^{\infty} \pi_{(k-n+1,0)}p(1-q_1) \Pr\{B>n-2\} \right) \Pr\{B=n-1|B>n-2\} \notag \\
={}& \left(\sum\nolimits_{k=n}^{\infty} \pi_{(k-n+1,0)}\right) p(1-q_1) \Pr\{B=n-1\} \notag \\
={}& p(1-q_1)q_{n-1}\left(\sum\nolimits_{k=1}^{\infty} \pi_{(k,0)}\right) \notag \\
={}& \frac{(1-q_1)\pi_{(1,0)}}{q_1} q_{n-1}
\end{align}

In order to obtain equation (19), notice that the last equation in (17) implies that
\begin{equation}
\sum\nolimits_{k=1}^{\infty} \pi_{(k,0)} = \frac{\pi_{(1,0)}}{pq_1}
\end{equation}

Substituting (19) into the third line of (17), we derive the following recursive formula for the probabilities $\pi_{(n,0)}$.
\begin{equation}
\pi_{(n,0)}=\pi_{(n-1,0)}(1-p) +\frac{(1-q_1)\pi_{(1,0)}}{q_1}q_{n-1}   \qquad  (n\geq2)
\end{equation}

Repeatedly using (21) yields following equations
\begin{align}
\pi_{(n,0)}&=\pi_{(n-1,0)}(1-p)+\frac{(1-q_1)\pi_{(1,0)}}{q_1}q_{n-1}  \notag \\
&=\left(\pi_{(n-2,0)}(1-p)+\frac{(1-q_1)\pi_{(1,0)}}{q_1}q_{n-2}\right)(1-p) + \frac{(1-q_1)\pi_{(1,0)}}{q_1}q_{n-1}  \notag \\
&=\pi_{(n-2,0)}(1-p)^2 + \frac{(1-q_1)\pi_{(1,0)}}{q_1}\left[(1-p)q_{n-2}+q_{n-1}\right]  \notag \\
& \qquad \qquad \qquad \qquad \qquad  \qquad \qquad   \vdots  \notag \\
&=\pi_{(1,0)}(1-p)^{n-1} + \frac{(1-q_1)\pi_{(1,0)}}{q_1} \left( \sum\nolimits_{j=0}^{n-2}(1-p)^{j}q_{n-1-j} \right)
\end{align}

Let $n=1$ in (22), we obtain the obvious equation $\pi_{(1,0)}=\pi_{(1,0)}$. Thus, the expression (22) is valid for all $n\geq 1$.

Since all the probabilities $\pi_{(n,m)}$ add up to 1, we have
\begin{align}
1&=\sum\nolimits_{m=0}^{\infty}\sum\nolimits_{n=m+1}^{\infty}\pi_{(n,m)}\notag \\
&=\sum\nolimits_{n=1}^{\infty}\pi_{(n,0)} + \sum\nolimits_{m=1}^{\infty}\sum\nolimits_{n=m+1}^{\infty}\pi_{(n,m)}  \notag \\
&=\sum\nolimits_{n=1}^{\infty}\pi_{(n,0)} +\sum\nolimits_{m=1}^{\infty}\sum\nolimits_{n=m+1}^{\infty}\pi_{(n-m,0)}p(1-q_1)\Pr\{B>m-1\}  \notag  \\
&=\sum\nolimits_{n=1}^{\infty}\pi_{(n,0)} + p(1-q_1)\sum\nolimits_{m=1}^{\infty}\left( \sum\nolimits_{n=1}^{\infty}\pi_{(n,0)} \right)\Pr\{B>m-1\} \notag \\
&=\left( \sum\nolimits_{n=1}^{\infty}\pi_{(n,0)} \right) \left[ 1+p(1-q_1)\sum\nolimits_{m=1}^{\infty}\Pr\{B>m-1\} \right]  \notag \\
&= \frac{\pi_{(1,0)}}{pq_1}\left[ 1+p(1-q_1)\sum\nolimits_{m=1}^{\infty}\Pr\{B>m-1\} \right]
\end{align}
from which the first probability $\pi_{(1,0)}$ can be determined as
\begin{equation}
\pi_{(1,0)}=\frac{pq_1}{1+p(1-q_1)\sum\nolimits_{m=1}^{\infty}\Pr\{B>m-1\}}
\end{equation}

Define
\begin{equation}
F(p,n)=\begin{cases}
1 & \qquad  (n=1)  \\
(1-p)^{n-1}+\frac{1-q_1}{q_1}\sum\nolimits_{j=0}^{n-2}(1-p)^{j}q_{n-1-j} &  \qquad  (n\geq2)
\end{cases}   \notag
\end{equation}

Thus, according to equation (22), we can write that
\begin{align}
\pi_{(n,0)}=\pi_{(1,0)}F(p,n) =\frac{pq_1F(p,n)}{1+p(1-q_1)\sum\nolimits_{m=1}^{\infty}\sum\nolimits_{l=m}^{\infty}q_l } \qquad (n\geq1)
\end{align}

The other probabilities $\pi_{(n,m)}$ can be determined by using recursive expression (18) as
\begin{align}
\pi_{(n,m)}&=\pi_{(n-m,0)}p(1-q_1) \Pr\{B>m-1\}  \notag \\
&=\pi_{(1,0)}F(p,n-m)p(1-q_1) \left(\sum\nolimits_{l=m}^{\infty}q_l \right) \notag \\
&=\frac{p^2q_1(1-q_1)F(p,n-m)\left(\sum\nolimits_{l=m}^{\infty}q_l\right)}{1+p(1-q_1)\sum\nolimits_{m=1}^{\infty}\sum\nolimits_{l=m}^{\infty}q_l }
\end{align}

So far, we completely solve the system of equations (17) and obtain all the stationary probabilities. This completes the proof of Theorem 2.

\section{Proof of Theorem 3}
In this appendix, we solve the system of stationary equations when the queue adopted in system is Ber/Geo/1/1 queue.

It is not difficult to obtain the solutions in Theorem 3 by using the general results (6)-(9). Firstly, from the equation (22) we have
\begin{align}
\pi_{(n,0)}&=\pi_{(1,0)}(1-p)^{n-1} + \frac{(1-q_1)\pi_{(1,0)}}{q_1}\left( \sum\nolimits_{j=0}^{n-2}(1-p)^jq_{n-1-j} \right)  \notag \\
&=\pi_{(1,0)}(1-p)^{n-1} +\frac{(1-\gamma)\pi_{(1,0)}}{\gamma}\left( \sum\nolimits_{j=0}^{n-2}(1-p)^j(1-\gamma)^{n-2-j}\gamma \right)  \notag \\
&=\pi_{(1,0)}\bigg[(1-p)^{n-1} + (1-\gamma)^{n-1}\sum\nolimits_{j=0}^{n-2}\left(\frac{1-p}{1-\gamma}\right)^j\bigg]   \notag \\
&=\pi_{(1,0)}\bigg[(1-p)^{n-1} + \frac{(1-p)^{n-1}(1-\gamma)-(1-\gamma)^n}{\gamma-p}\bigg]  \notag \\
&=\pi_{(1,0)}\frac{(1-p)^n-(1-\gamma)^n}{\gamma-p}
\end{align}

The first probability $\pi_{(1,0)}$ can be determined from equation (24) as
\begin{align}
\pi_{(1,0)}&=\frac{pq_1}{1+p(1-q_1)\sum\nolimits_{m=1}^{\infty}\Pr\{B>m-1\}}  \notag \\
&=\frac{p\gamma}{1+p(1-\gamma)\sum\nolimits_{m=1}^{\infty}\sum\nolimits_{k=m}^{\infty}(1-\gamma)^{k-1}\gamma}  \notag \\
&=\frac{p\gamma}{1+p(1-\gamma)\sum\nolimits_{m=1}^{\infty}(1-\gamma)^{m-1}}  \notag \\
&=\frac{p\gamma}{1+p(1-\gamma)(1/\gamma)} \notag \\
&=\frac{p\gamma^2}{p+\gamma-p\gamma}
\end{align}

Combining (27) and (28), we obtain
\begin{equation}
\pi_{(n,0)}=\frac{p\gamma ^2\left[(1-p)^n - (1-\gamma)^n \right]  }{(p+\gamma -p\gamma)(\gamma -p)}  \qquad (n\geq2)
\end{equation}

Notice that expression (29) is also valid for the case $n=1$.

The other probabilities $\pi_{(n,m)}$ can be found by calculating equation (18) directly. For $n>m\geq1$, we have
\begin{align}
\pi_{(n,m)}&=\pi_{(n-m,0)}p(1-q_1)\Pr\{B>m-1\}  \notag \\
&=\pi_{(n-m,0)}p(1-\gamma)\sum\nolimits_{k=m}^{\infty}(1-\gamma)^{k-1}\gamma \notag \\
&=\pi_{(n-m,0)}p(1-\gamma)^m    \notag \\
&=\frac{(p\gamma)^2\left[(1-p)^{n-m}(1-\gamma)^m-(1-\gamma)^n  \right]}{(p+\gamma-p\gamma)(\gamma-p)}
\end{align}

So far, all the probabilities are determined. This completes the proof of Theorem 3.

\section{Proof of Corollary 3}
The average discrete AoI $\overline{\Delta}^{(d)}_{Ber/Geo/1/1}$ is calculated as
\begin{align}
\overline{\Delta}^{(d)}_{Ber/Geo/1/1}&= \sum\nolimits_{n=1}^{\infty}n\Pr\{\Delta=n\}  \notag \\
&= \frac{p(1-p)\gamma ^3}{(p+\gamma -p\gamma)(\gamma -p)^2}\sum\nolimits_{n=1}^{\infty}\left[n(1-p)^n - n(1-\gamma)^n \right]
- \frac{(p\gamma)^2}{(p+\gamma -p\gamma)(\gamma -p)}\sum\nolimits_{n=1}^{\infty}n^2(1-\gamma)^n  \notag \\
&= \frac{p(1-p)\gamma ^3}{(p+\gamma -p\gamma)(\gamma -p)^2}\frac{(p+\gamma-p\gamma)(\gamma-p)}{p^2\gamma^2}
- \frac{(p\gamma)^2}{(p+\gamma -p\gamma)(\gamma -p)}\frac{p^2(1-\gamma)(2-\gamma)}{\gamma^3}  \notag \\
&= \frac{(1-p)\gamma}{p(\gamma-p)}-\frac{p^2(1-\gamma)(2-\gamma)}{(p+\gamma-p\gamma)(\gamma-p)\gamma} \notag \\
&= \frac{(1-\rho_d\gamma)\gamma}{\rho_d\gamma(\gamma-\rho_d\gamma)}-\frac{\rho_d^2\gamma^2(1-\gamma)(2-\gamma)}{(\rho_d\gamma+\gamma-\rho_d\gamma^2)(\gamma-\rho_d\gamma)\gamma}   \\
&= \frac{1}{\gamma}\left( \frac{1}{\rho_d}+\frac{1-\gamma}{1-\rho_d} - \frac{\rho_d(1-\gamma)}{1-\rho_d}+ \frac{\rho_d(1-\gamma)}{1+(1-\gamma)\rho_d} \right)   \notag \\
&= \frac{1}{\gamma}\left((1-\gamma) + \frac{1}{\rho_d} + \frac{\rho_d}{\frac{1}{1-\gamma}+\rho_d}\right)
\end{align}
where in (29) we substitute $\rho_d=p/\gamma$.

It is easy to see that
\begin{equation}
\overline{\Delta}^{(d)}_{Ber/Geo/1/1}\to \overline{\Delta}_{M/M/1/1}=\frac{1}{\mu}\left(1 + \frac{1}{\rho} + \frac{\rho}{1+\rho}\right)    \notag
\end{equation}
when $(1-\gamma)\to1$, i.e., when the time slot is short enough.

This completes the proof of Corollary 3.

\end{document}